\documentclass[%
 reprint,
 amsmath,amssymb,
 aps,amsthm,
]{revtex4-2}
\usepackage{physics}
\usepackage{graphicx}
\usepackage{dcolumn}
\usepackage{bm}

\usepackage[utf8]{inputenc}
\usepackage{enumerate}
\usepackage{mathtools}
\usepackage{algorithm,algorithmic}
\usepackage{blindtext}
\usepackage{graphicx}
\graphicspath{ {./images/} }
\usepackage{epsfig}
\usepackage{sidecap}
\usepackage{hyperref,stackrel}
\usepackage[normalem]{ulem}
\usepackage{color}
\usepackage{enumerate}
\usepackage{bm}
\usepackage{amsthm,amsmath}
\usepackage{amsfonts}
\usepackage{amssymb,wasysym}
\usepackage{graphicx}
\usepackage{enumitem}
\usepackage{lmodern}
\usepackage{braket}
\usepackage[mathscr]{euscript}

\usepackage{multirow}
\usepackage [english]{babel}
\usepackage [autostyle, english = american]{csquotes}
\usepackage{caption}
\usepackage{subcaption}

\theoremstyle{definition}
\newtheorem{definition}{Definition}[section]
\newtheorem{prop}{Proposition}
\newtheorem{theorem}{Theorem}[section]

\newtheorem{lemma}[theorem]{Lemma}

\makeatletter
\newcommand*{\rom}[1]{\expandafter\@slowromancap\romannumeral #1@}
\makeatother
\begin{document}

\preprint{APS/123-QED}

\title{Multiparty Entanglement Routing in Quantum Networks}

\author{Vaisakh Mannalath}\email{vaisakhmannalath@gmail.com}

\author{Anirban Pathak}
 \email{anirban.pathak@jiit.ac.in}
\affiliation{%
 Department of Physics and Materials Science \& Engineering, Jaypee Institute of Information Technology, A-10, Sector 62, Noida, UP-201309, India
  }%

\date{\today}

\begin{abstract}
Distributing entanglement among multiple users is a fundamental problem in quantum networks, requiring an efficient solution. In this work, a protocol is proposed for extracting maximally entangled (GHZn) states for any number of parties in quantum networks of arbitrary topology. It is based on the graph state formalism and requires minimal assumptions on the network state. The protocol only requires local measurements at the network nodes and just a single qubit memory per user. Existing protocols on bipartite entanglement routing are also improved for specific nearest-neighbor network architectures. To this end, the concept of majorization is utilized to establish a hierarchy among different paths in a network based on their efficacy. This approach utilizes the symmetry of the underlying graph state to obtain better-performing algorithms.

\end{abstract}


\maketitle

\section{\label{sec:Introduction}Introduction}Point-to-point secure quantum communication has been achieved quite successfully with optical fibers   \cite{CZL+21}, and in free space,   \cite{ACS+21}. However, direct quantum communication is limited by errors and losses incurred during transmission   \cite{PLO+17}. In order to overcome these inherent limitations and to establish connections over long distances, approaches based on entanglement swapping are employed   \cite{SSR+11, ATL15}. Over recent years, a set of such protocols have been experimentally implemented with increasing distances, and success rates   \cite{YCZ+08,LZY+19,LTM+21}. In fact, with the help of a satellite, intercontinental quantum communication has already been performed between China and Austria   \cite{LCH+18}. Apart from pushing the boundaries of what is possible regarding maximum distance   \cite{CZL+22} and communication rates   \cite{CZC+21}, a natural step forward would be to consider multiparty scenarios. Naturally, the community hopes to develop more complicated networks involving multiple nodes leading to a quantum internet   \cite{S17}. Many protocols in the multiparty scenario have been proposed for tasks like quantum secret sharing    \cite{HBB99,GKB+07,TZG01,XLD+04}, quantum voting   \cite{HZB+06,TSP16,MTP+22}, and quantum conference key agreement   \cite{MGK+20,PHG+21, HJP20,BTS+18}. The study of such quantum network protocols is an active field with promising applications. 

A fundamental requirement for realizing a quantum internet is to develop algorithms for managing the entanglement present in the network and, thus, to distribute entangled states among two or more specific nodes (users)   \cite{LLY+19,MMG19,BSK+20}. This easy-to-state algorithm designing problem is fundamental and challenging, and under the chosen condition, it leads to a set of related problems of interest. For example, in Ref.   \cite{DHW20}, the authors investigated whether a given graph state can be transformed into a set of Bell states between specific network nodes using operations restricted to single-qubit Clifford operations, single-qubit Pauli measurements, and classical communication. They showed that this specific problem is \texttt{NP-Complete}. This result highlights the difficulty of the problem at hand and the crucial need to devise better-performing protocols, at least for some specific instances relevant to multipartite schemes.

Before we proceed further, it would be apt to note that the effectiveness of quantum networks in performing tasks unachievable in the classical world (in classical networks) is not restricted to secure quantum communication only. Quantum network is also helpful in distributed quantum computing   \cite{BR03,LBK05}, clock synchronization  \cite{C00,JAD+00, GLM01,NSN+22}, and many other applications. Most of these applications require the ability to distribute entanglement among two or more network nodes located far away. Motivated by this, a few schemes for generating entanglement among specific network nodes have been proposed in the recent past   \cite{LLC21,MBF10,PPE+22,PKT+19}, and some of them have been implemented experimentally  \cite{HKM+18,KRH+17,ZCB+21,SNN+20}. Interestingly, the optimal or maximally efficient protocol for remote entanglement generation in a quantum network has yet to be discovered. This motivated us to look at the possibility of designing a more efficient algorithm for the entanglement distribution among the nodes of a network.\\

A helpful tool used in the study of quantum networks is the notion of graph states   \cite{S03,HEB04}. They have been employed to realize several tasks in quantum information processing, including quantum metrology   \cite{SM20}, quantum error correcting codes   \cite{SW01} and one-way quantum computing \cite{RB01}. Furthermore, a strong interplay between the graph theory and quantum entanglement is known, and the same has been investigated from various perspectives   \cite{HEB04,SS19,AMD+20}. Graph states are generated in a network when the nodes, sharing maximally entangled pairs with nearby nodes, perform suitable entanglement-generating operations locally. Graph states have been studied extensively in the context of quantum networks  \cite{EKB16,PD19}, with much of the research focused on generating them in a quantum network with varying assumptions   \cite{MMG19,CC12,PWD18}.

A general method for extracting maximally entangled states with two or three parties in connected networks was presented in    \cite{HPE19}. They consider manipulating an already generated graph state to accommodate future communication requests. Compared to other methods, such as the algorithm described in  \cite{SMI+16} which requires large amounts of quantum memories, this approach was more advantageous regarding the memory required for the repeater stations. In the graph state formalism,  the maximally entangled states shared by two (three) parties are represented by line graphs, with two (three) vertices, up to some local operations. Hence, to establish a maximally entangled state between two (three) nodes, we could perform sequential entanglement swapping measurements on a path connecting all two (three) nodes. The protocol proposed in   \cite{HPE19} is quite similar to such repeater-based protocols, albeit more efficient. However, this simple approach will not work for four or more nodes since the corresponding graphs are not line graphs \cite{HEB04}. Moreover, many paths can connect the two nodes, even for the simplest two-party scenario. The protocol in  \cite{HPE19} does not provide a way to evaluate these different possibilities and pick the right one.

In our work, we define a new protocol for extracting maximally entangled states for any number of parties. The protocol only requires local measurements performed by the network users with access to a single qubit memory. In order to achieve this, we extensively use graph-theoretic tools in the graph state formalism of quantum networks. Moreover, we improve upon the results of \emph{Ref.}   \cite{HPE19} by providing a more efficient routine for establishing connections between distant nodes of a network. We use the concept of majorization   \cite{MOA79} to establish a hierarchy among different paths in a network based on their efficacy. This concept utilizes the symmetry of the underlying graph state to obtain better-performing algorithms.

The rest of the paper is organized as follows; in Section \ref{sec:prelim}, we briefly introduce graph states and graph theory tools that we use.  
Then, in Section \ref{sec:ghztstate}, we state and prove the theorem concerning multipartite state generation and demonstrate several examples. In Section \ref{sec:gridgraphs}, we consider a class of nearest neighbour graphs and improve upon the existing entanglement routing protocols for those specific cases. In Section \ref{sec:discussion} we conclude with our remarks and possible future research questions.\\

\begin{figure*}
    \centering
    \includegraphics[width=0.7\linewidth]{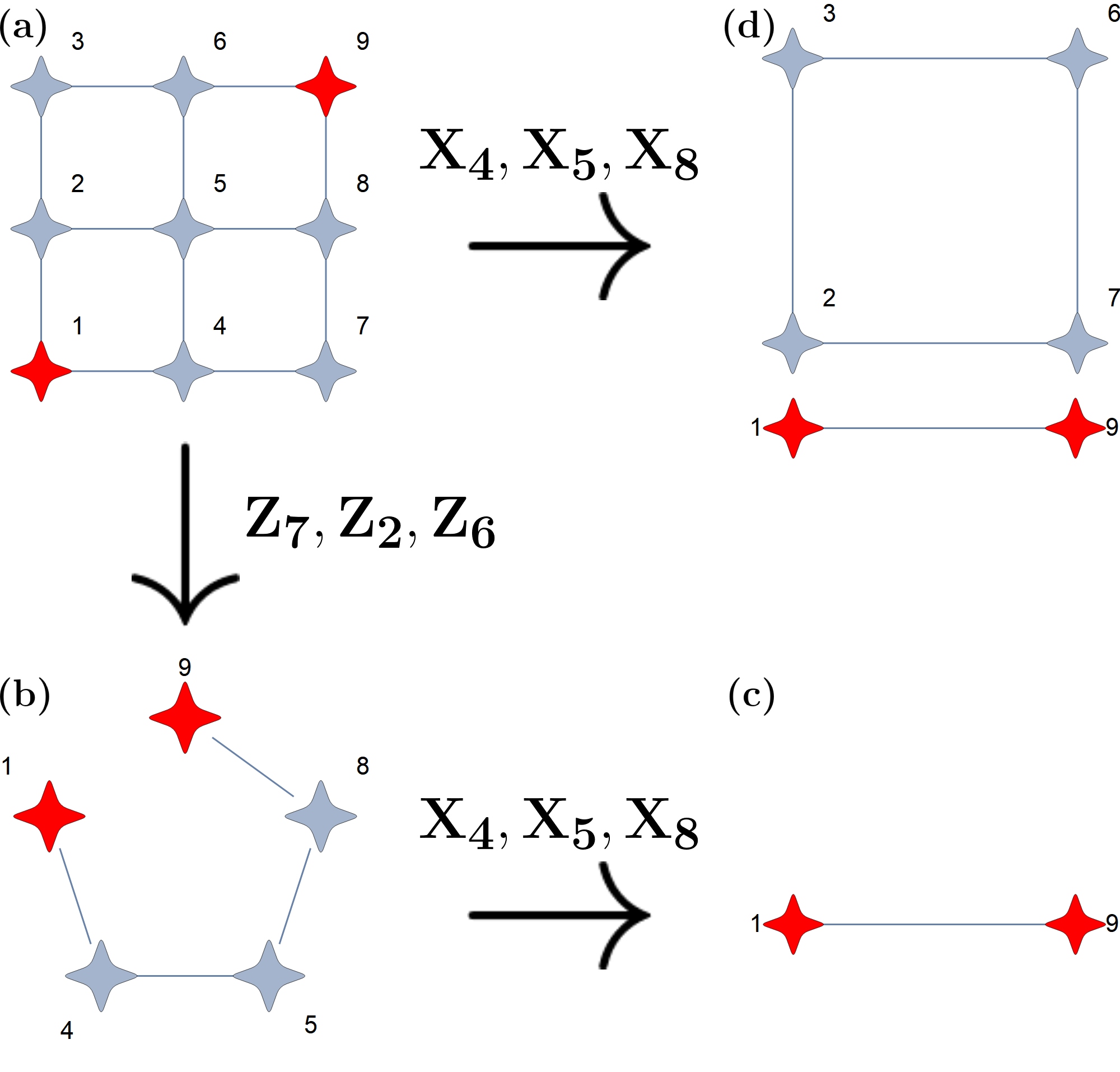}
    \caption{(Color online) Repeater protocol and $X$ protocol   \cite{HPE19}
    \textbf{(a)}$3\times3$ square grid. The Objective is to establish a Bell pair between the vertices $1,9$. \textbf{(b)-(c)} Repeater protocol. First, we isolate a repeater line $1,4,5,8,9$ with $Z$ measurements on the neighbourhood vertices $7,2,6$. Then, intermediate vertices along the repeater line are measured in $X$ basis sequentially, resulting in a Bell pair. \textbf{(c)}. \textbf{(d)} $X$ protocol. Instead of isolating the repeater lines, we directly $X$ measure the vertices $4,5,8$. The next step would be $Z$-measurement of all neighbourhood vertices of $1,9$, but it is not necessary in this specific case. Thus, the $X$ protocol requires fewer measurements compared to the repeater protocol and yields an additional four-node state, along with the required Bell pair.}
    \label{repeatandx}
\end{figure*}
\begin{figure*}
    \centering
    \includegraphics[width=0.9\linewidth]{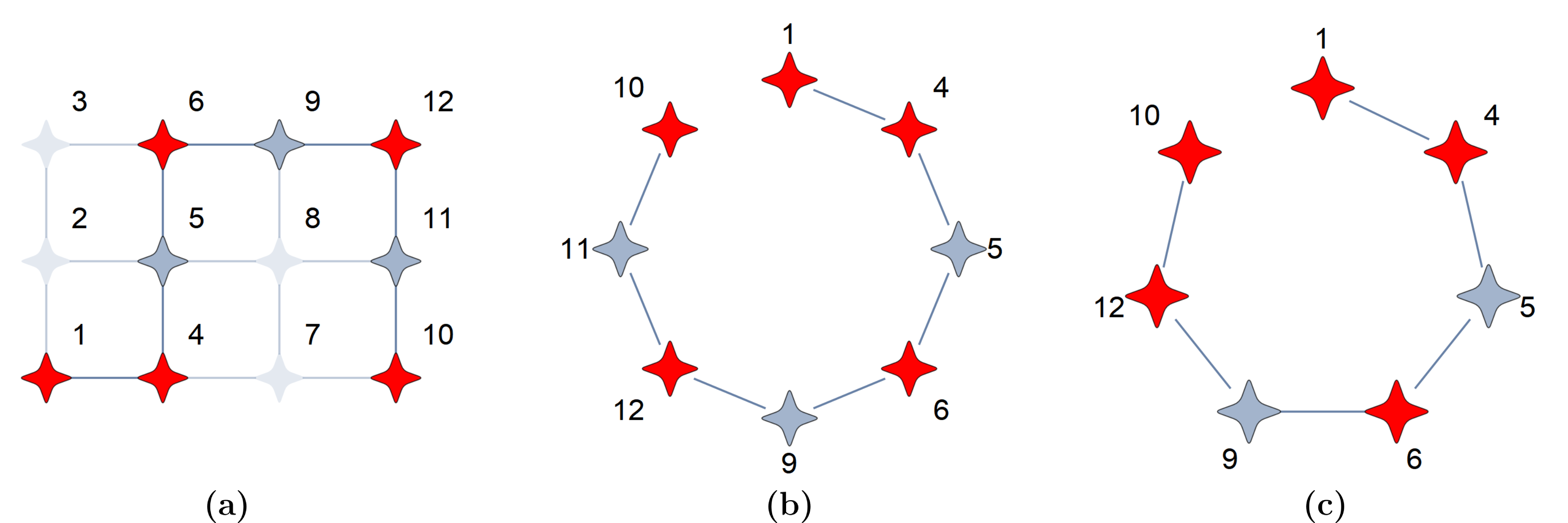}
    \caption{(Color online) Isolating the desired repeater line. \textbf{(a)} A $3\times 4$ grid network. The highlighted path connects the vertices $1,4,6,12,10$, which are to be part of the final GHZ state. Z-measurement on the vertices $3,2,8,7$ isolates this path from the rest of the graph. \textbf{(b)} The isolated path. The vertex $11$ is not required for the protocol, and we can remove it using an X-measurement. \textbf{(c)} The repeater line as required by \autoref{theo1}. It contains the five nodes of the final GHZ state and extra nodes between the intermediate nodes $12,6,4$.}
    \label{isorepeat}
\end{figure*}
\begin{figure*}
    \centering
    \includegraphics[width=0.9\linewidth]{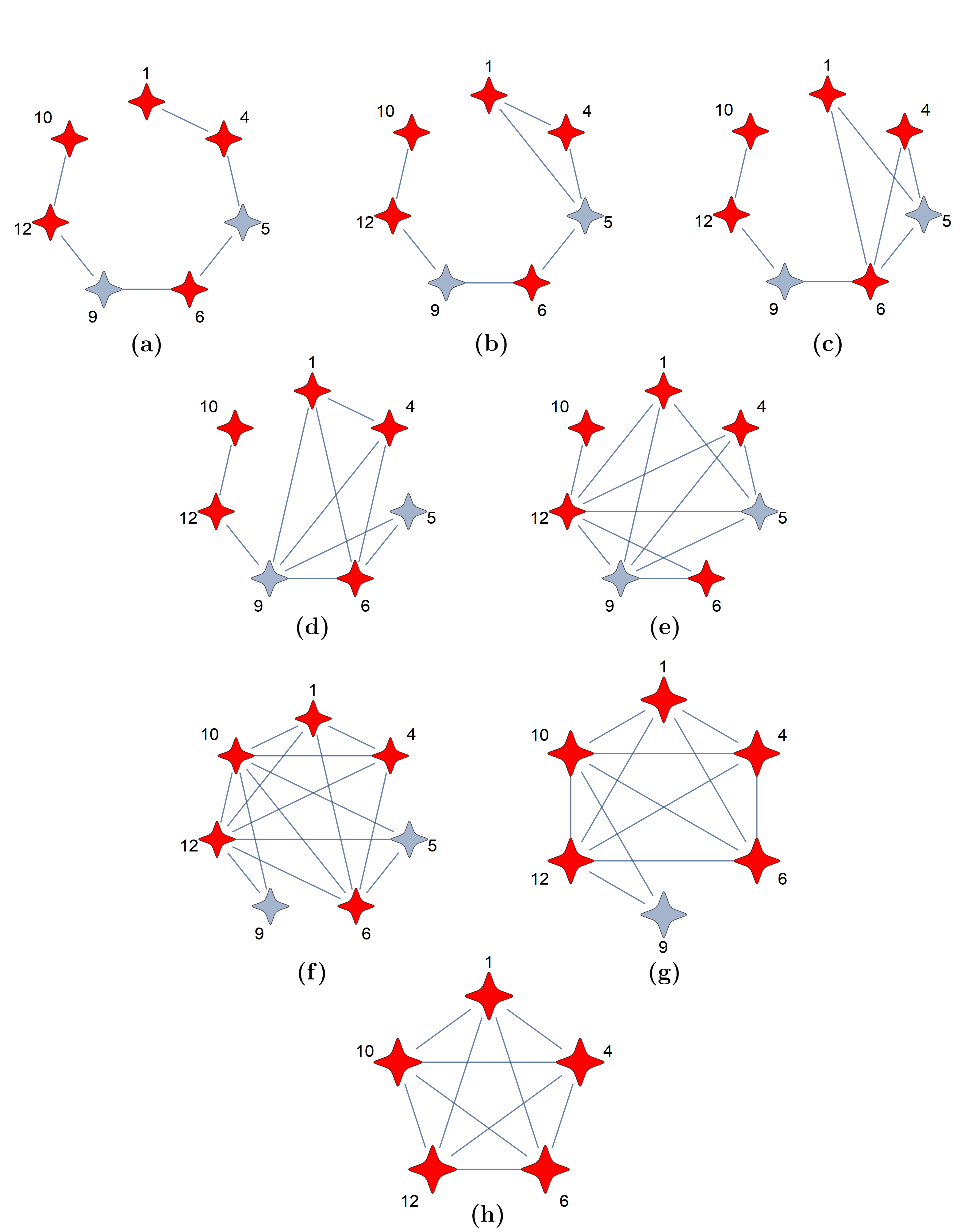}
    \caption{(Color online) Extracting GHZ5 state \textbf{(a)} The isolated repeater line from \autoref{isorepeat}. \textbf{(b)-(f)} Sequential $LC$ on vertices $4,5,6,9,12$. \textbf{(g)-(h)} $Z$-measurement of vertices $5,9$. }
    \label{seqLC}
\end{figure*}
\begin{figure*}
    \centering
    \includegraphics[width=0.85\linewidth]{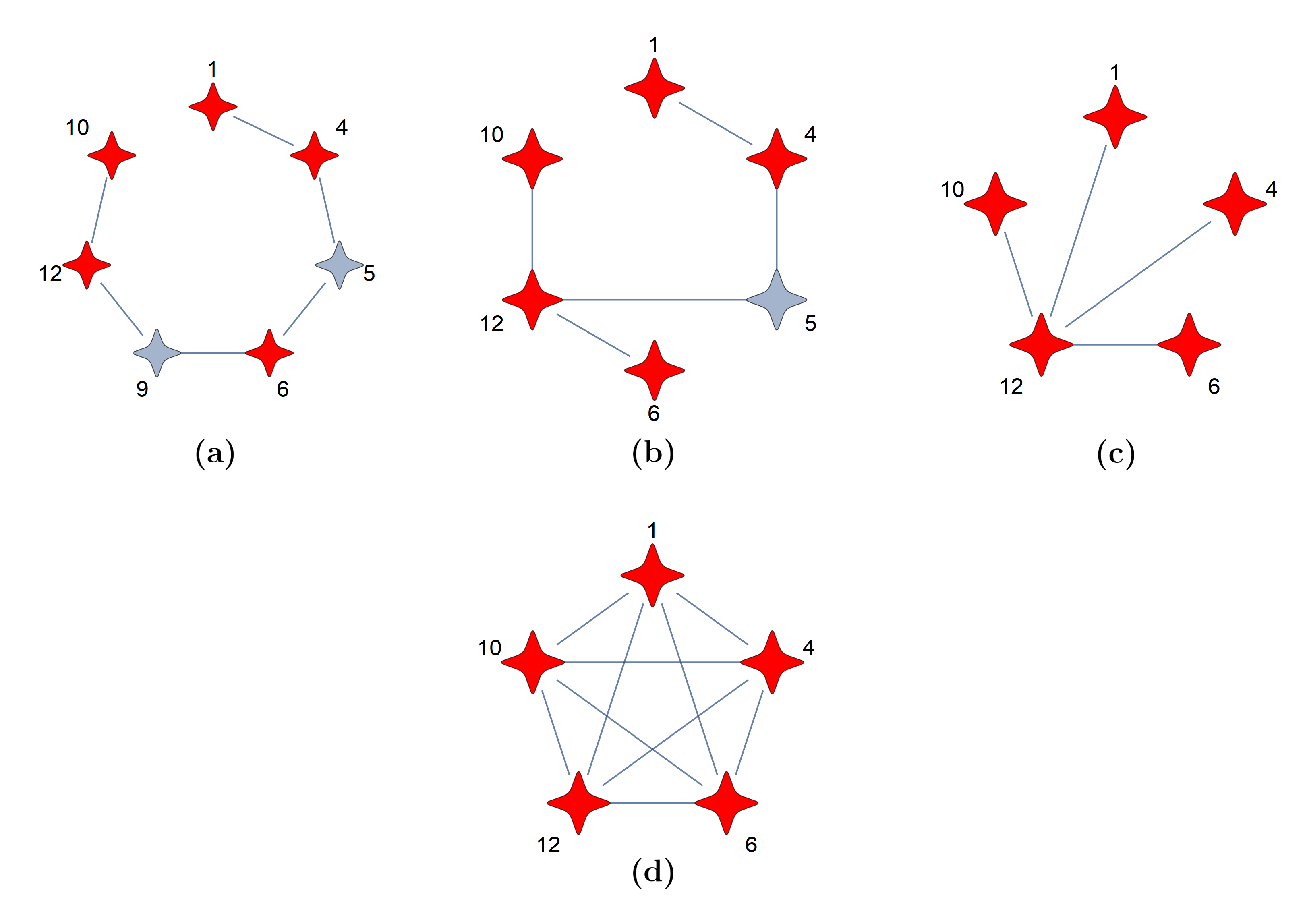}
    \caption{(Color online) Extracting GHZ5 state \textbf{(a)} The isolated repeater line from \autoref{isorepeat}. \textbf{(b)} $X$-measurement on vertex $9$. \textbf{(c)} $X$-measurement on vertex $5$. \textbf{(d)} The final state obtained by performing $LC$ on vertex $12$.}
    \label{seqx}
\end{figure*}
\section{\label{sec:prelim}Preliminaries}

An undirected finite graph $G=(V, E)$ is defined by a set of vertices $V \subsetneq \mathbb{N}$ and a set $E \subseteq V \times V$ of edges. A simple graph is a graph without any loop (an edge that connects a vertex with itself)  and multiple edges connecting the same pair of vertices. The set of all vertices having a shared edge with a given vertex $a$ is called the neighborhood of $a$ and is denoted by $N_a$.\\

\begin{definition}
 (Vertex Deletion): Deleting a vertex $v$ results in a graph where the vertex $v$ and all the edges connected to it are removed.
 $$G-v=\left(V \backslash v,\{e \in E: e \cap v=\varnothing\}\right)$$.
 \end{definition}
\begin{definition}
(Local complementation): A local complementation $LC_v$ is a graph operation specified by a vertex $v$, taking a graph $G$ to $LC_v(G)$ by replacing the neighborhood of $v$ by its complement.
$$
LC_v(G)=\left( V,E\Delta K_{N_v}\right),
$$
where $K_{N_v}$ is the set of edges of the complete graph
on the vertex set $N_v$ and $E \Delta K_{N_v}=(E \cup K_{N_v})-(E \cap K_{N_v})$ is the symmetric difference.
\end{definition}
Local complementation acts on the neighbourhood of a vertex by removing edges if they are present and
adding missing edges, if any.
\begin{definition}
(Vertex-minor): A graph $H$ is called a vertex-minor of $G$ if a sequence of local complementations and vertex-deletions maps $G$ to $H$.
\end{definition}
   The simple graph $G\left( V,E\right)$ defined so far is a mathematical entity, but in the quantum world, we can associate a pure quantum state $|G\rangle$ with it, called a graph state. A Graph state is defined on a Hilbert space $\mathcal{H}_V=\left(\mathbb{C}^{2}\right)^{\otimes V}$. Specifically, each vertex in $V$ is assigned a qubit in the state $|+\rangle=(|0\rangle+|1\rangle) / \sqrt{2}$.  Subsequently, a controlled-$Z$ operation is applied to a pair of qubits sharing an edge to construct the graph state $|G\rangle$ associated with the graph $G\left( V,E\right)$   \cite{HEB04}. Thus, a graph state is defined as follows; 
$$|G\rangle:=\prod_{(i, j) \in E} C Z_{i, j}|+\rangle^{\otimes V}.$$

Local Clifford operations on graph states defined above can be represented using local complementations on the corresponding graph   \cite{VDD04}. Local Pauli measurements on the graph states can be represented using local complementations, and vertex deletions   \cite{HEB04}. We can visualize the role played by the Pauli measurements as follows.

  \begin{prop}
 ($Z$-measurement) Measurement of a qubit, corresponding to the vertex $v$, in the $Z$-basis is represented by the vertex deletion of $v$.
 $$Z_v(G)=G-v
 $$
 \end{prop}
 \begin{prop}
 ($Y$-measurement) Measurement of a qubit, corresponding to the vertex $v$, in the $Y$-basis is represented by,
 $$Y_v(G)=Z_vLC_v(G)
 $$
 \end{prop}
 \begin{prop}
 \label{prop:xmeas}
 ($X$-measurement) Measurement of a qubit, corresponding to the vertex $v$, in the $X$-basis is represented by
 $$X_v(G)=LC_wZ_vLC_vLC_w(G)
 $$
 where $w\in N_v$.
 \end{prop} 

\begin{definition}
(Repeater Line) A path for which  $E(V,V)=\{(v_1,v_2),\cdots,(v_{n-1},v_{n})\}$ for $V=\{v_1,\cdots,v_n\}$ is called a repeater line.
\end{definition}

\section{\label{sec:ghztstate}GHZ states}

We have already stated a set of definitions and propositions which allow us to initiate discussion on the solution of a problem stated as follows: Given a graph state, how can we efficiently share entangled states between multiple parties that are not connected via physical channels ? To answer this question, let us first look at the simplest example of connecting two distant nodes through a shared Bell pair. The straightforward solution would be to find the shortest path connecting the two nodes with the smallest combined neighbourhood,  performing $Z$ measurements on the nodes that do not lie on the path and sequential $X$-measurement on all the intermediate nodes on the path (See \autoref{repeatandx}(a)-(c) for an illustrative example). Here, the initial $Z$ measurements would isolate a repeater line between the initial vertices, and the $X$-measurement on a vertex has the simple action of connecting its neighbours and deleting the vertex itself. It is easy to see then how this protocol gives the desired result. Since a GHZ3 state can also be represented using a line graph, the same protocol as above can also be applied there. First, we find a path connecting all three nodes between which the GHZ3 state is to be distributed, isolate it from the rest of the graph using $Z$ measurements, and remove intermediate vertices using $X$ measurements. Since generating a repeater line is an essential step of these protocols, we will call them as \emph{repeater protocol} from now onward.

 An improvement over the repeater protocol described above was provided in   \cite{HPE19}. The authors provided a more efficient method for establishing entangled states than the protocol described above. Here, the efficiency of a protocol refers to the total number of measurements required to enact the protocol. The lesser the number of measurements needed, the greater the connections left in the graph, which could be used in subsequent rounds for further generation of specific entangled states. This new protocol leverages the properties of $X$ measurements on graph states and is aptly called the \emph{$X$ protocol} (see \autoref{repeatandx} (a),(d)). For the Bell pair generation between two nodes, the protocol proceeds as follows: find the shortest path between two nodes with the minimum combined neighbourhood, perform $X$ measurements on all the intermediate vertices, and subsequent $Z$-measurement of all the neighbouring vertices of the two vertices. The $X$ measurements connect the two nodes, as in the repeater protocol. The difference here is that the two nodes acquire additional neighbours during the $X$-measurement step, which increases the number of $Z$ measurements required in the next step. The authors proved \cite{HPE19} that the total number of measurements required for the $X$ protocol was nevertheless lesser than the repeater protocol. They also applied the $X$ protocol to generate GHZ3 states among three specific vertices.

 Now let us consider the problem of establishing GHZ states shared among more than $3$ parties. As explained before, the abovementioned methods cannot be generalized easily to the multiparty case. We now prove that one can extract GHZn from a connected graph, given the graph satisfies a vertex minor condition. The said requirement will likely be satisfied by graph states shared over future quantum network architectures. The protocol will then be analyzed using specific examples to showcase their benefits and to compare with previously known results.

\begin{figure*}
    \centering
    \includegraphics[width=0.9\linewidth]{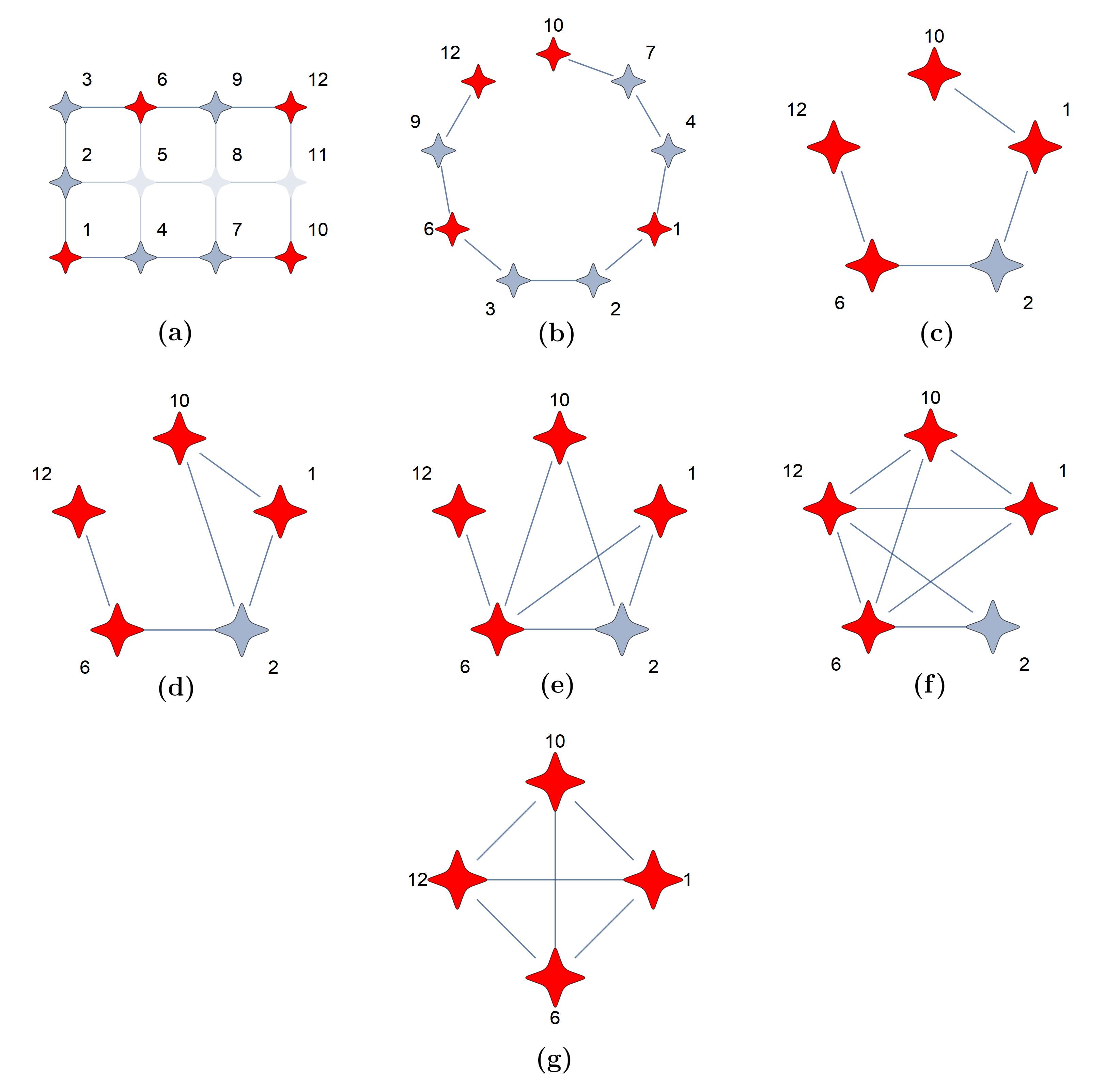}
  \caption{(Color online) Extracting GHZ4 state; example given in  \cite{HPE19}. \textbf{(a)}  Identify a suitable path connecting the vertices $12,6,1,10$ that are to be part of the final state. \textbf{(b)} Isolate the path using $Z$ measurements on vertices $5,8,11$. \textbf{(c)} Appropriate measurements on vertices $9,3,4,7$ to reach this five vertex graph. \textbf{(d)-(f)} Sequential $LC$ on $1,2,6$. \textbf{(g)} $Z$-measurement of $2$ yields the desired GHZ4 state.}
    \label{prev4ghz}
\end{figure*}

\begin{figure*}
    \centering
\includegraphics[width=\linewidth]{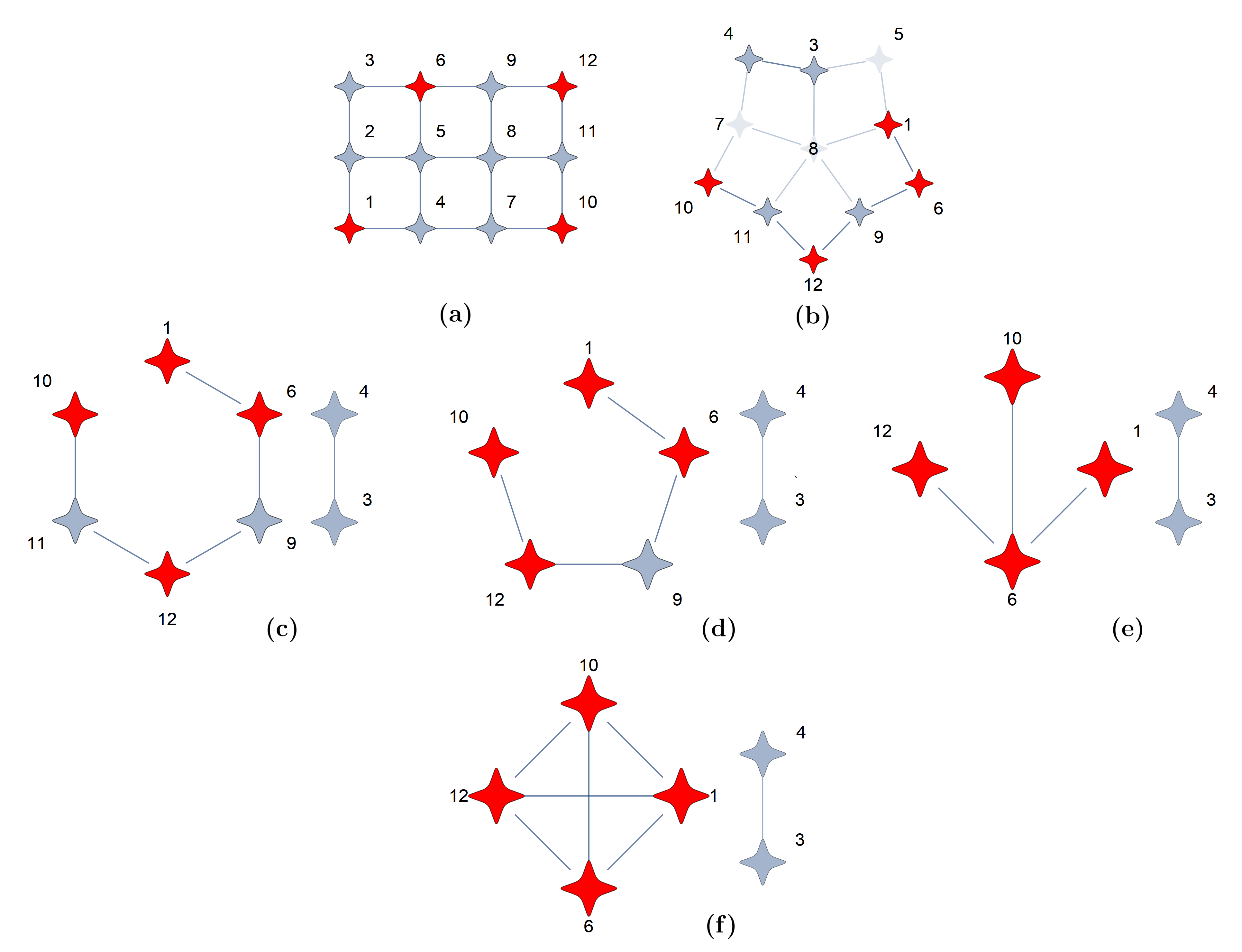}
    \caption{(Color online) Extracting GHZ4 state and a Bell pair \textbf{(a)} The same underlying graph and vertices are to be connected, as in \autoref{prev4ghz}.\textbf{(b)} $X$-measurement on $2$ converts the grid state to one of its vertex-minors \textbf{(c)} $Z$ measurements on $7,8,5$. \textbf{(d)} $X$-measurement on $11$. \textbf{(e)} $X$-measurement on $9$. \textbf{(f)} $LC$ on $6$.}
    \label{our4ghz}
\end{figure*}
\begin{theorem}
\label{theo1}
(Extraction of GHZn states) It is possible to extract an $n$-partite GHZ state from a graph state $\ket{G}$ when the underlying graph has a
repeater line as vertex-minor, connecting all $n$ nodes of the final GHZ state with an extra node in between every pair of $n-2$ intermediate nodes.
\end{theorem}
Given such a repeater line, the GHZ state can be obtained by performing sequential local complementations on all vertices except the two at both ends and subsequent $Z$-measurement of all the neighbouring vertices not part of the final state.
\begin{proof}
The proof is obtained by keeping track of the neighbourhood of the nodes throughout the measurement process. For simplicity, let us assume the case where the required repeater line is isolated from the rest of the graph through some  $Z$ measurements. 
Let us denote vertices along the path as $v_i$, vertices at even positions as $e_i$ and  at odd positions by $o_i$(except those at both ends). We have $e_i=v_{2i}$ and $o_i=v_{2i+1}$
\begin{lemma}
\label{LClemma}
After performing $LC$ on every vertex up to $v_i$ in an isolated repeater line, the neighbourhood of  $v_{i+1}$ is given by $N_{v_{i+1}}=\{v_1,\cdots,v_i,v_{i+2}\}$.
\end{lemma}
\begin{proof}
 This can be easily proved through induction. After $LC$ on $v_2$, $N_{v_3}=\{v_1,v_2,v_4\}$. Assume lemma to be true for $n$; after $LC$ on $v_n$, $N_{v_{n+1}}=\{v_1,\cdots,v_{n},v_{n+2}\}$. At this point $N_{v_{n+2}}=\{v_{n+1},v_{n+3}\}$. $LC$ on $v_{n+1}$ gives us \begin{eqnarray}
 N_{v_{n+2}}&= & N_{v_{n+1}}\cup\{v_{n+1},v_{n+3}\}\setminus v_{n+2}\nonumber\\
 &=&\{v_1,\cdots,v_{n+1},v_{n+3}\},\nonumber
 \end{eqnarray} proving the lemma. This implies that $LC$ on  $v_i$ connects $v_{i+1}$ to every vertex before it in the repeater line.
\end{proof}
Take any two vertices $e_i,e_j;i<j$. Since the choice of $i,j$ is arbitrary, it suffices to prove that they remain connected at the end of the protocol. We can now look at how the edge connection between these vertices develops throughout the protocol. Note that initially, $e_i$ and $e_j$ are not connected. There exists at least one odd vertex in between them. The two vertices are connected for the first time when one performs $LC$ on the odd vertex $o_{j-1}$ (Lemma \ref{LClemma}). Next $LC$ is to be performed on $e_j$; which does not affect the connection. However, this operation connects $o_j$ to $e_i$. Thus, the next $LC$, performed on $o_j$ removes the edge between $e_i$ and $e_j$. Since now $e_{j+1}$ is connected to $e_i$ and $e_j$ and $e_i$, $e_j$ are not connected to each other, $LC$ on $e_{j+1}$ rebuilds the edge between $e_i$ and $e_j$. This pattern is followed by the rest of the protocol, where $LC$ on odd vertices removes the edge and $LC$ on even vertices adds the edge. Since the sequence of $LC$'s ends with an even vertex, owing to the specific construction of the repeater line, $e_i$, $e_j$ remains connected at the end of $LC$ operation. $Z$ measurements on odd vertices will not affect the connectivity of even vertices. Since this applies for any $i,j$ with $i<j$, all $n$ vertices are connected at the end of this protocol, and we have the desired GHZ state.
\end{proof}
\begin{lemma}
\label{xprotocolproof}
(Generalized $X$ protocol)Assuming the  repeater line required by \autoref{theo1} exists, performing $X$ measurement on the odd vertices and subsequent $Z$-measurement of all the
neighbouring vertices not part of the final state yields the desired GHZ state
\end{lemma}
\begin{proof}
We will show how this generalized $X$ protocol is equivalent to the protocol used to prove \autoref{theo1}. As before, we assume that the repeater line is isolated from the rest of the graph. Such a repeater line requires, by \autoref{theo1}, at least $2n-3$ vertices for constructing the GHZn state. Lets number the vertices sequentially and represent $LC$, $Z$, $X$ measurements on vertex $v_i$ as $LC_i$, $Z_i$, $X_i$, respectively. The protocol presented above can be then represented by,
$$Z_3Z_5\cdots Z_{2n-7} Z_{2n-5}LC_{2n-4}LC_{2n-5}\cdots LC_3LC_2.
$$
Note that $Z_i$ commutes with $LC_j$, since a measurement removes the node and all the edges connected to it. Rearranging, we get,
\resizebox{\linewidth}{!}
{$LC_{2n-4}(Z_{2n-5}LC_{2n-5}LC_{2n-6})\cdots(Z_5LC_5LC_4)(Z_3LC_3LC_2)$}.
The terms in the bracket looks similar to $X-$measurement performed on node $v_i$, besides a term corresponding to the final $LC$ of the neighbour node $v_{i-1}$,
$$
X_i\equiv LC_{i-1}Z_iLC_iLC_{i-1}.
$$
A key observation here is that, after the $Z_i$ measurement, out of all the $n$ vertices to be included in the final GHZ state, $v_{i-1}$ is only connected to $v_{i+1}$. Even if we consider the case where the repeater line is not isolated, local complementation on $v_{i-1} $ will only change the edge connectivity of its neighbours, which will be deleted anyway. So it does not matter if we perform the local complementation or not. The final $LC$, $LC_{2n-4}$ can also be ignored since it's just a local operation on the final GHZ state itself.
Hence we can rewrite the protocol to be,
$$
X_{2n-5}X_{2n-7}\cdots X_5X_3.
$$
\end{proof}

In \autoref{isorepeat}, we have shown an example of isolating the desired repeater line from an underlying $3\times4$ grid-graph to subsequently distill a  GHZ5 state composed of the nodes $1,4,6,12,10$. Note that depending on the underlying graph and the distribution of vertices, such a repeater line may or may not exist. In this case, the repeater line composed of the vertices $1,4,5,6,9,12,11,10$ satisfies the conditions laid out in \autoref{theo1}, specifically, the existence of an extra node between the intermediate nodes $4,6,$ and $12$. Note that the node $11$ is unnecessary for the protocol and can be removed after isolating the repeater line. After we isolate the repeater line, there are two equivalent ways to generate the final GHZ state. In \autoref{seqLC}, sequential $LC$'s are applied to the repeater line vertices. This connects every vertex that's supposed to be part of the final state to each other. Removing all the unnecessary vertices in the final step yields the required state. In \autoref{seqx} $X$ measurement are carried out on the vertices that are not part of the final state. Both methods are equivalent and result in the same final state.

In the above theorems and examples, we have presented a case where the repeater line is first isolated from the graph state to apply the protocol. This is unnecessary, and one can perform the protocol directly on the repeater line while embedded within the graph state and subsequently isolate the final state from the rest of the graph using accordingly chosen measurements. This is quite similar to the contrast between the repeater protocol and $X$ protocol . The authors \cite{HPE19} had shown that operating on the repeater line before isolating the final state requires fewer measurements than if we go the other way around. We can prove similar results for our generalized version. The proof of the following lemma is given in the appendix.

\begin{lemma}
\label{xbeforiso}
The generalized $X$ protocol performed before isolating the repeater line requires, at most many measurements as the one where the repeater line is isolated first.
\end{lemma}
\begin{proof}
Proof of this theorem is given in the appendix.
\end{proof}

Extracting Bell pairs and GHZ3 states with local operations is always possible in a connected graph   \cite{HPE19}. Our theorem reflects this possibility since the requirement of "an extra node in between every pair of $n-2$ intermediate nodes" is trivially satisfied for $n=2,3$. In   \cite{HPE19}, a sufficient criterion for extracting GHZ4 states was also provided. This special case ($n=4$) of our generalized protocol is presented in \autoref{prev4ghz}. A crucial part of the protocol is the generation of the repeater line connecting the final four vertices $10,1,6,12$, with an extra node between the two intermediate vertices $1,6$ (\autoref{prev4ghz} \textbf{(c)}). Thus, it satisfies the conditions of \autoref{theo1}, and the protocol proceeds by performing sequential $LC$ and $Z$ measurements, the same as in our protocol. 

 We note here that there is an even more efficient way of extracting the GHZ4 state for this specific example. A downside of using the protocol in \autoref{prev4ghz} is that it removes the possibility of extracting more entangled states from the same graph. Here, we show how this protocol can be modified to extract an extra Bell pair and the GHZ state from the same graph state (\autoref{our4ghz}). Instead of isolating the repeater line, we perform suitably chosen $X$ measurement to modify the graph state to enable simultaneous extraction of multiple states. Essentially, it converts the original state \autoref{our4ghz}\textbf{(a)} to one of its vertex-minors  \autoref{our4ghz}\textbf{(b)}, such that it enables us to distill an extra Bell pair without disturbing the GHZ4 extraction protocol.

\begin{figure*}
    \centering
    \includegraphics[width=0.9\linewidth]{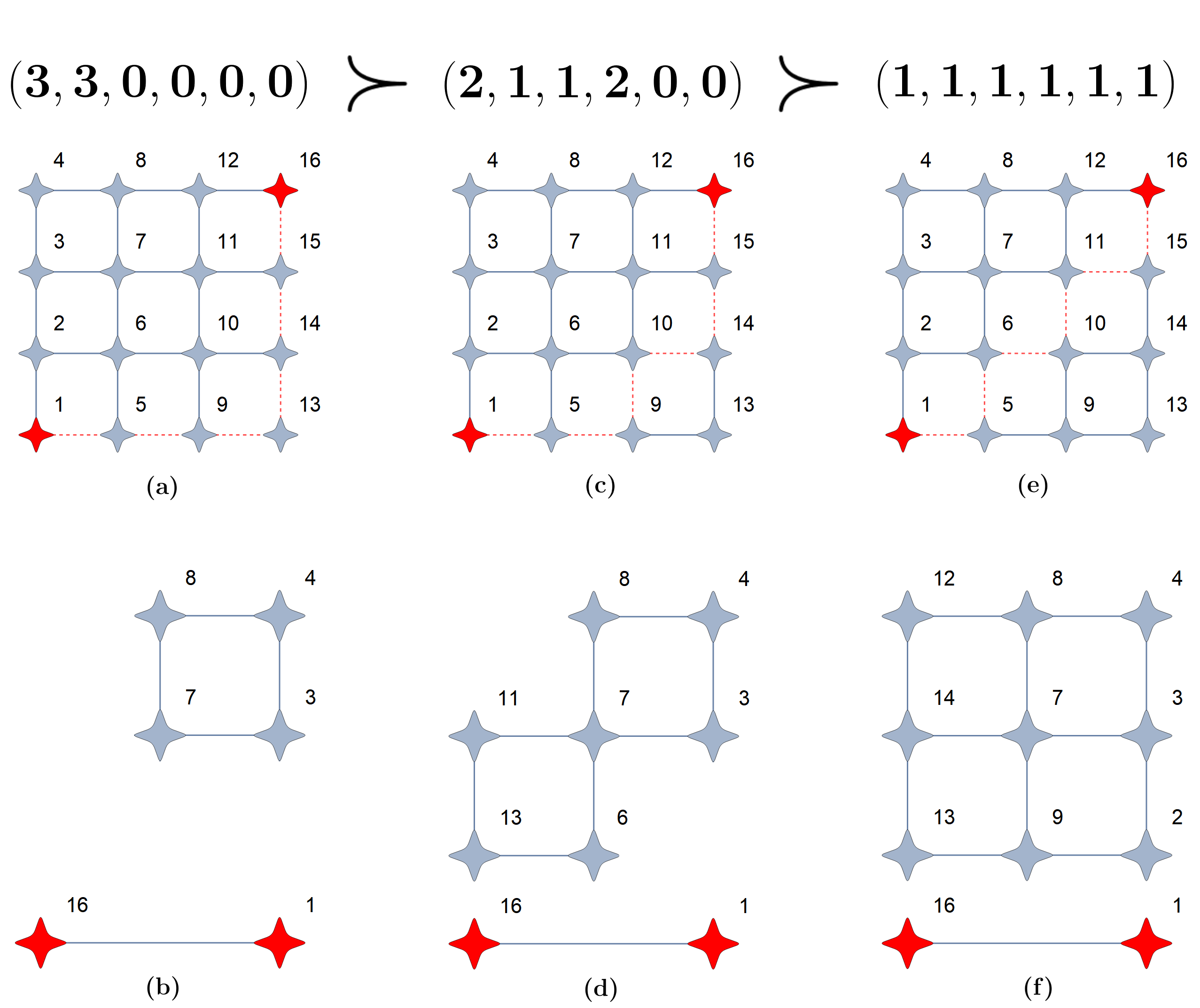}
    \caption{(Color online) Entanglement routing paths and their end-products. \textbf{(a)-(b)} $X$ protocol is performed along the path $1,5,9,13,14,15,16$. The path is equivalently represented by the vector $(3,3,0,0,0,0)$. In addition to the desired Bell pair between $1,16$, we also obtain a $2\times 2$ grid state. \textbf{(c)-(d)} A different path that requires a lesser amount of measurements than the previous path. This is evident from the higher number of connected vertices left in the graph. \textbf{(e)-(f)} The optimal path to perform the $X$ protocol. The path vector $(1,1,1,1,1,1)$ corresponding to this path is majorized by every other path. By \autoref{majortheo}, this path maximizes the amount of entanglement left in the graph.}
    \label{majorize}
\end{figure*}

\section{\label{sec:gridgraphs}Grid Graphs}

 Grid graphs belong to an important class of network architecture known as nearest-neighbour networks. Such networks are relevant to practical quantum communication since they connect nearest neighbours via physical links. Consequently, the quantum information only has to travel a short distance, thus minimizing transmission losses and errors. Nearest neighbour networks like rings, lines, and grid graphs have been studied under different scenarios, especially in communication bottlenecks. In the simplest scenario, a bottleneck arises in a network when two pairs of nodes intend to share a Bell pair over a common edge. It has been shown in   \cite{HDE+22} that ring and line networks cannot overcome such bottlenecks, whereas it has been shown to be possible in grid graphs in the case of a \emph{``butterfly network"}   \cite{HPE19,LOW10}.

This section will show how symmetry in grid graphs can be utilized in the entanglement routing problem. We will explore the simplest routing task of establishing a Bell pair between two distant nodes (with no direct physical link) in a connected graph. A naive solution to this would be the repeater protocol. We have already mentioned that a more efficient approach was discussed in   \cite{HPE19}, which we have referred to as the $X$ protocol. The authors proved that for the shortest with the minimum combined neighbourhood path, $X$ protocol requires fewer measurements than the repeater protocol.

However, there is still an ambiguity left on the \emph{choice} of the path to make, as the minimum combined neighbourhood path is not always unique. To illustrate this point, in \autoref{majorize}, we show all the different paths (with minimum length and combined neighbourhood) one can have given a graph and a pair of nodes. The theorems in   \cite{HPE19} just tell us that, given any of these paths, the $X$ protocol would perform better than the repeater protocol; it does not provide any information on the ideal path to choose. It is unclear whether all the paths are equivalent or if one performs better. Thus, a set of open questions remain, and in what follows, we will address these questions.  

In the following, we define a property to these paths, enabling us to define a ``better" path qualitatively. We use the concept of majorization, an ordering relation on real vectors. Majorization has already seen extensive applications in quantum information, including entanglement theory  \cite{N99} and formulation of resource theories   \cite{GMN+15}. Here, we will use this tool to judge the entanglement routing paths in a quantum network qualitatively. We associate a vector to all the feasible paths, referred to as the \emph{path vector}. Given a set of such path vectors, we can impose an ordering relation based on the majorization relation. This would enable us to infer the best path to choose for entanglement routing.

The first step towards defining the path vector is associating a sense of direction within the grid graph. This is done easily by viewing the grid graph on a cartesian plane and locating every vertex on the coordinates $(x,y);x,y\in \mathbb{N}$, with adjacent vertices separated by unit distance on either of the coordinates. Every edge in this setting will be parallel to the $x$ or the $y$ axis. This implies that every path on a grid graph can be represented using a sequence of edges along the $x$ and $y$ directions. Since we are only concerned with the shortest path between two nodes, without loss of generality, we can restrict ourselves to positive $x$ and $y$. We now have the necessary ingredients to define a path vector.

 Suppose we have two vertices $a$ and $b$, with the coordinates $(x,y)$ and $(x',y')$. Define $D_x=|x'-x|$ and $D_y=|y'-y|$. $D=D_x+D_y$ is the number of edges in the shortest path connecting $a,b$. A path vector $\mathbf{s}$ given by $(s_{x_1},s_{y_1},s_{x_2},s_{y_2}\cdots,s_{x_{D_x}},s_{y_{D_y}})$ defines a shortest path between $a$ and $b$ iff $s_{x_1}+s_{x_2}+\cdots+s_{x_{D_x}}=D_x$ and $s_{y_1}+s_{y_2}+\cdots+s_{y_{D_y}}=D_y$. Here, $s_{x_i}$ is the number of consecutive edges the path covers in the $x$ direction, after which it changes direction and covers $s_{y_i}$ number of edges in the $y$ direction. Hence every shortest path will have a unique vector. We will represent the path with path vector $\mathbf{s}$  as S.

\begin{definition}
(Majorization) Given $\mathbf{s,t}\in\mathbb{R}^d$ we say that $\mathbf{t}$ is majorized by  $\mathbf{s}$ (written as $\mathbf{s} \succ \mathbf{t}$) iff

\begin{equation}
\sum_{i=1}^k s_i^{\downarrow} \geq \sum_{i=1}^k t_i^{\downarrow} \quad \text { for } k=1, \ldots, d-1\nonumber
\end{equation}
\begin{equation}
\sum_{i=1}^d s_i=\sum_{i=1}^d t_i,\nonumber\\
\end{equation}
where
$\mathbf{s}^{\downarrow} \in \mathbb{R}^d $ is the vector $\mathbf{s}$ with the same components but sorted in descending order. Similar definition holds for $\mathbf{t}^{\downarrow}$. 
\end{definition}

\begin{theorem}
Given path vectors $\mathbf{s}$, $\mathbf{t}$ and $
\mathbf{s} \succ \mathbf{t}$, $X$ protocol along T requires at most as many measurements as the one along S.
\label{majortheo}
\end{theorem}
\begin{proof}
Proof of this theorem is given in the appendix.
\end{proof}

 For example, in \autoref{majorize}, we have shown three different paths on a $4\times4$ grid graph to establish a Bell pair between the vertices $1,16$. The length and combined neighbourhood of all paths are the same. We have shown the end result of the $X$ protocol and the corresponding path vector for each path. From the figure, it is clear that, given any pair of paths, one with the majorized path vector produces better results in the $X$ protocol. 
\section{\label{sec:discussion}Discussions}
This paper provides a protocol for extracting multipartite entangled states from quantum networks using just local measurements. Our protocol is applicable on a shared network state with a single qubit per location and is hence favorable in terms of the repeater memory required compared to previous protocols. It relies on constructing a  repeater line connecting the final vertices and local measurements on the nodes in between. Another important contribution of this work is the key observation that the protocol does not mandate the repeater line be removed from the underlying graph state, which reduces the number of measurements required for the overall protocol. Note that this protocol could still be applied if we allow for multiple qubit memories per user. We have also analyzed nearest neighbour networks and shown how their structure and symmetry can help us devise better protocols for entanglement routing.\\
This work is expected to open a window for an investigation into a plethora of new problems. For example, it would be worthwhile to explore whether the protocol could be adapted to incorporate all-photonic quantum repeaters   \cite{ATL15}. Such memory-less repeaters were introduced to establish entanglement between two distant nodes, where the repeaters would use cluster states to connect neighbouring stations. Measurements on such cluster states would eventually connect the distant nodes and establish entanglement between them. Since our protocol also requires just local $X$ measurements along a repeater line, it would be interesting to see how the all-photonic repeaters could be modified to fit our protocol. Proof-of-principle experiments based on this concept have been demonstrated successfully in recent years   \cite{LZY+19,HIM+19}, and its resource efficiency with matter-memories has been characterized   \cite{HBE21}. This makes us optimistic about the practical realization of our work in a broad domain.\\
Throughout this article, we have considered the idealistic scenario; zero-channel losses, perfect measurements, etc. In the bipartite routing scenario, the authors of   \cite{PKT+19} have calculated the entanglement generation rate under non-ideal circumstances in a network. They have shown that by taking advantage of multiple routing paths in a network, the rate of entanglement generation can supersede that is achievable using a linear repeater chain. Recently, the rates were shown to be \emph{distance independent} if GHZ projective measurements are allowed at the repeater nodes   \cite{PPE+22}. Note that these results are for generating a Bell pair among network users. A natural line of research would be to consider the multipartite rates using our protocol. \\
Recently, \emph{Ref.}\cite{JHE+22} proposed a protocol for anonymous conference key agreement among any three participants of a linear network. The users need only to share Bell pairs with their neighbours and avoid the necessity of a central server sharing multipartite states. The network users perform local measurements to extract the GHZ3 state among the participants, which can subsequently be used for key agreement. Interestingly, the protocol introduced in this work can be used to generate GHZn states from a linear network. Further, It would be interesting to see if we can extend the concepts laid out in  \cite{JHE+22} to propose a similar protocol for $n$ parties.
\section*{Acknowledgment}
VM, AP acknowledges support from the QUEST scheme of the Interdisciplinary Cyber-Physical Systems (ICPS) program of the Department of Science and Technology (DST), India, Grant No.: DST/ICPS/QuST/Theme-1/2019/14 (Q80). AP also acknowledges support from the QUEST scheme of Interdisciplinary Cyber-Physical Systems (ICPS) program of the Department of Science and Technology (DST), India, Grant No.: DST/ICPS/QuST/Theme-1/2019/6 (Q46). They also thank Kishore Thapliyal for their interest and feedback on the work.

\newpage
\newpage
\section{Appendix}
\subsection{Proof of Lemma \ref{xbeforiso}}
\textbf{Lemma \rom{3}.4.} The generalized $X$ protocol performed before isolating the repeater line requires, at most many measurements as the one where the repeater line is isolated first.
\begin{proof}
Let us count the number of measurements required for the protocol in both cases. Assume that we are trying to generate a GHZn state. By \autoref{theo1}, an appropriate repeater line should contain at least $2n-3$ vertices.
\subsubsection{With Isolation}
In this case one requires $|N_{v_1}^0\cup\cdots\cup N_{v_{2n-3}}^0|-(2n-3)$ measurements for isolating the repeater line. The first term includes the neighbourhood vertices of all the vertices on the repeater line. However, this also includes the  $2n-3$ vertices themselves. Since we are not measuring them at this stage, the second term accounts for this. After isolating, we need to perform the protocol on the repeater line. This involves $(n-3)$ $X$ measurements of all the extra nodes. Thus, in total we require $|N_{v_1}^0\cup\cdots\cup N_{v_{2n-3}}^0|-n$ measurements.
\subsubsection{Without Isolation}
Here, we first perform the $X$-measurments of all the intermediate nodes. This step takes $n-3$ measurements, the same as before. After this step, we need to isolate the $n$-party state from the rest of the graph. This requires $|N^{n-3}_{v_1}\cup N^{n-3}_{v_2}\cup\cdots\cup N_{v_{2n-4}}^{n-3}\cup N_{v_{2n-3}}^{n-3}|-n$ measurements. The term in this expression accounts for the updated neighbourhood vertices of the $n$ vertices that are part of the final state after the $X$ measurements on the repeater line. Since this also includes the $n$
 vertices themselves, we subtract it. Thus, the protocol requires a total of $|N^{n-3}_{v_1}\cup N^{n-3}_{v_2}\cup\cdots\cup N_{v_{2n-4}}^{n-3}\cup N_{v_{2n-3}}^{n-3}|-3$ measurements.\\

 Now, all we need to prove is that
 \small
\begin{equation}
    |N^{n-3}_{v_1}\cup N^{n-3}_{v_2}\cdots N_{v_{2n-4}}^{n-3}\cup N_{v_{2n-3}}^{n-3}|\leq |N_{v_1}^0\cup\cdots\cup N_{v_{2n-3}}^0|-(n-3).
\label{setrelat}
\end{equation}
\normalsize
A crucial observation here is that $X$ measurements on the repeater line do not change the combined neighbourhood of the repeater line, except for the number of deleted vertices. This implies that,
$$N^{n-3}_{v_1}\cup N^{n-3}_{v_2}\cdots N_{v_{2n-4}}^{n-3}\cup N_{v_{2n-3}}^{n-3}\subset N_{v_1}^0\cup\cdots\cup N_{v_{2n-3}}^0.$$
The subset relation is proper since the LHS does not contain the $n-3$ vertices in RHS that were $X$ measured. This proves \autoref{setrelat} and hence the lemma.
\end{proof}

\subsection{Proof of \autoref{majortheo}}

\textbf{Theorem \rom{4}.1.} Given path vectors $\mathbf{s}$, $\mathbf{t}$ and $
\mathbf{s} \succ \mathbf{t}$, $X$ protocol along T requires at most as many measurements as the one along S.

\begin{proof}

Given two path vectors, one being majorized by the other implies that the path contains fewer consecutive edges along a particular direction. The majorized path thus involves more changes in direction. A direct consequence is that a higher number of alternating vertices along the path share a common neighbour. We show how the expressions for the total number of measurements derived in the supplementary information of   \cite{HPE19} relate to the neighbourhood intersections of vertices along the path.\\
In the following, we denote by  $N_{v_i}^{(t)}$ the neighborhood of node  $v_i$ after the  $t^{t h}$  Pauli measurement is performed on the initially given graph state. Briefly recalling, in the $X$ protocol, we first measure all the intermediate nodes along the path in $X$ basis and subsequently $Z$ measure all the neighbours of the connected pair. If the path contains $l$ nodes, we require $(l-2)$ number of $X$ measurements and some number of $Z$ measurements (depending on the number of neighbours of the connected pair). In the supplementary information of    \cite{HPE19}, the authors have derived an expression for the total number of $Z$ measurements, given by,

\begin{equation}
\label{neighnumb}
\begin{aligned}
&N_{v_1}^{(l-2)}=\left(N_{v_{l-1}}^{(0)} \cup N_{v_1}^{(l-4)}\right) \backslash\left(N_{v_{l-1}}^{(0)} \cap N_{v_1}^{(l-4)}\right) \\
&N_{v_l}^{(l-2)}=\left\{v_1\right\} \cup\left(N_{v_l}^{(0)} \cup N_{v_1}^{(l-3)}\right) \backslash\left(N_{v_l}^{(0)} \cap N_{v_1}^{(l-3)}\right)
\end{aligned}
\end{equation}

	Since they only considered the shortest path while proving this result, the path length is fixed, and the number of intermediate measurements required is the same for the different paths. Thus, we need to focus only on the number of measurements required to isolate the endpoints $v_1$ and $v_l$. Let us first consider the neighbourhood of the first vertex $v_1$ after all the $X$ measurements,
	
	$$N^{(l-2)}_{v_1}=\left(N^{(0)}_{v_{l-1}}\cup N^{(l-4)}_{v_{1}}\right) 
	\setminus \left(N^{(0)}_{v_{l-1}}\cap N^{(l-4)}_{v_{1}}\right)$$
 It is clear that greater the intersection between the terms $N^{(0)}_{v_{l-1}}$ and $N^{(l-4)}_{v_{1}}$, the lesser the cardinality of $N^{(l-2)}_{v_1}$.

Exapnding $N_{v_1}^{(l-4)}$ using \autoref{neighnumb}, we obtain

\small
\begin{eqnarray}
\label{1vertexpa}
N^{(l-2)}_{v_1}&=&N^{(0)}_{v_{l-1}}\cup\left( \left(N_{v_{l-3}}^{(0)}\cup N_{v_1}^{(l-6)}\right) 
	\setminus \left(N_{v_{l-3}}^{(0)}\cap N_{v_1}^{(l-6)}\right)\right)\setminus \nonumber\\ & \hspace{4em}&N^{(0)}_{v_{l-1}}\cap\left( \left(N_{v_{l-3}}^{(0)}\cup N_{v_1}^{(l-6)}\right)
	\setminus \left(N_{v_{l-3}}^{(0)}\cap N_{v_1}^{(l-6)}\right)\right).\nonumber
	\end{eqnarray}
\normalsize	
Consider the second half of the above expression.
	$$
N^{(0)}_{v_{l-1}}\cap\left( \left(N_{v_{l-3}}^{(0)}\cup N_{v_1}^{(l-6)}\right) 
	\setminus \left(N_{v_{l-3}}^{(0)}\cap N_{v_1}^{(l-6)}\right)\right)
	$$
	Using the recursive formula \autoref{neighnumb}, one sees that $N_{v_1}^{(l-6)}$, when expanded, contains terms of the form $N_{v_{l-5}}^{(0)},N_{v_{l-7}}^{(0)}\cdots$.
	Since we are restricting ourselves to the shortest paths in a grid graph, none of those terms will have non-empty intersections with $N^{(0)}_{v_{l-1}}$. Thus, the above expression reduces to,
	\begin{equation}
	\label{1vert1term}
	    	N^{(0)}_{v_{l-1}}\cap N_{v_{l-3}}^{(0)}
	\end{equation}

Now, for the first half of \autoref{1vertexpa},
	$$
	N^{(0)}_{v_{l-1}}\cup\left( \left(N_{v_{l-3}}^{(0)}\cup N_{v_1}^{(l-6)}\right) 
	\setminus \left(N_{v_{l-3}}^{(0)}\cap N_{v_1}^{(l-6)}\right)\right).
	$$
	Since  	$N^{(0)}_{v_{l-1}}$ have no intersection with the term $\left(N_{v_{l-3}}^{(0)}\cap N_{v_1}^{(l-6)}\right)$, we can rewrite the above expression  as,
	$$
\left(	N^{(0)}_{v_{l-1}}\cup N_{v_{l-3}}^{(0)}\cup N_{v_1}^{(l-6)} \right)
	\setminus \left(N_{v_{l-3}}^{(0)}\cap N_{v_1}^{(l-6)}\right).
	$$
	Since $N_{v_{l-3}}^{(0)}$ only intersects with $N_{v_{l-5}}^{(0)}$ out of all the terms in the expansion of $N_{v_1}^{(l-6)}$, this becomes,
	$$
\left(	N^{(0)}_{v_{l-1}}\cup N_{v_{l-3}}^{(0)}\cup N_{v_1}^{(l-6)} \right)
	\setminus \left(N_{v_{l-3}}^{(0)}\cap N_{v_{l-5}}^{(0)}\right).
	$$ 
	One can keep applying the same type of argument to terms of the form $N_{v_1}^{(l\neq0)}$, and together with \autoref{1vert1term} we get

	\begin{eqnarray}
	N^{(l-2)}_{v_1}&=& \left(	N^{(0)}_{v_{l-1}}\cup N_{v_{l-3}}^{(0)}\cup N_{v_{l-5}}^{(0)}\cup\cdots\right)\nonumber \\ &\hspace{4em}&\setminus\left(N_{v_{l-1}}^{(0)}\cap N_{v_{l-3}}^{(0)}\right)\setminus \left(N_{v_{l-3}}^{(0)}\cap N_{v_{l-5}}^{(0)}\right)\cdots.\nonumber
	\end{eqnarray}
	Similiarly for $	N^{(l-2)}_{v_l}$,
	
	\begin{eqnarray}
	N^{(l-2)}_{v_l}&=&\left(	N^{(0)}_{v_{l}}\cup N_{v_{l-2}}^{(0)}\cup N_{v_{l-4}}^{(0)}\cup\cdots\right)\nonumber \\ &\hspace{4em}&\setminus\left(N_{v_{l}}^{(0)}\cap N_{v_{l-2}}^{(0)}\right)\setminus \left(N_{v_{l-2}}^{(0)}\cap N_{v_{l-4}}^{(0)}\right)\cdots.\nonumber	\end{eqnarray}
Thus the total $Z$ measurements required becomes,
\small
\begin{eqnarray}
|N^{(l-2)}_{v_1}\cup	N^{(l-2)}_{v_l}|&=&|N^{(0)}_{v_{l}}\cup	N^{(0)}_{v_{l-1}}\cup
N_{v_{l-2}}^{(0)}\cup
 \cdots \cup N_{v_{1}}^{(0)}|-\nonumber\\
 &&|\left( N_{v_{l}}^{(0)}\cap N_{v_{l-2}}^{(0)} \right)\cup \left(N_{v_{l-1}}^{(0)}\cap N_{v_{l-3}}^{(0)}\right)\cup\nonumber\\
 &&\cdots\cup\left( N_{v_{1}}^{(0)}\cap N_{v_{3}}^{(0)} \right)|.\nonumber
\end{eqnarray}
\normalsize
The first term in the above expression is the combined neighourhood of the path and the second term is the sum of intersections between alternate vertices. For the shortest path in a grid graph, the expression is minimized when the number of neighbourhood intersections of vertices on the path is maximized. 
\end{proof}

\bibliographystyle{naturemag}
\bibliography{network.bib}

\end{document}